%% file: MCSP.tex
\documentclass{article}
\usepackage{graphics,latexsym}

\usepackage{graphicx,epsfig,color}

\usepackage{graphicx}

\usepackage{enumerate}

\usepackage{amsfonts}

\input{mac90}

\usepackage{enumerate}

\newcommand{\MCSP}{{\rm MCSP}}

\newcommand{\tally}{{\rm TALLY}}

\newcommand{\ZPP}{{\rm ZPP}}

\newcommand{\ZEP}{{\rm ZPEXP}}

\newcommand{\littleo}{{\rm o}}

\newcommand{\PrTime}{{\rm BTIME}}

\newcommand{\ZPrTime}{{\rm ZTIME}}

\newcommand{\streaming}{{\rm Streaming}}

\newcommand{\ex}{{\rm exp}}

\newcommand{\MKtP}{{\rm MKtP}}

\newcommand{\bigO}{{\rm O}}

\newcommand{\Circuit}{{\rm Circuit}}

\newcommand{\mod}{{\rm mod}}

\newcommand{\Pad}{{\rm Pad}}

\newcommand{\GF}{{\rm GF}}

\begin{document}

\date{}

\title{Hardness of Sparse Sets  and Minimal Circuit Size Problem}
\author{Bin Fu\\
 		Department of Computer Science,\\
        University of Texas Rio Grande Valley,        Edinburg, TX 78539, USA.\\
        	bin.fu@utrgv.edu
} 
\maketitle

\begin{abstract}
We 
 study the magnification of
hardness of sparse sets in nondeterministic time complexity classes
on a randomized streaming model. One of our results shows that if
there exists a $2^{n^{o(1)}}$-sparse set in $\NTIME(2^{n^{o(1)}})$
that does not have any randomized streaming algorithm with
$n^{o(1)}$ updating time, and $n^{o(1)}$ space, then
$\NEXP\not=\BPP$, where a $f(n)$-sparse set is a language that has
at most $f(n)$ strings of length $n$. We also show that if $\MCSP$
is  $\ZPP$-hard under polynomial time truth-table reductions, then
$\EXP\not=\ZPP$.


\end{abstract}


\section{Introduction}
Hardness magnification  has been intensively studied in the recent
years~\cite{OliveiraSanthanam18,ChenJinWilliams19,McKayMurrayWilliams19,OliveiraPichSanthanam19}.
A small lower bound such as $\Omega(n^{1+\epsilon})$ for one problem
may bring a large lower bound such as super-polynomial lower bound for
another problem. This research is closely related to Minimum Circuit
Size Problem (\MCSP) that is to determine if a given string of
length $n=2^m$ with integer $m$ can be generated by a circuit of
size $k$. For a function $s(n):\mathbb{N}\rightarrow \mathbb{N}$, $\MCSP[s(n)]$ is
that given a string $x$ of length $n=2^m$, determine if there is a
circuit of size at most $s(n)$ to generate $x$. This problem has
received much attention in the recent
years~\cite{AllenderHoldenKabanets17,AllenderHirahara17,OliveiraSanthanam18,HitchcockPavan15,HiraharaWatanabe16,HiraharaSanthanam17,HiraharaOliveiraSanthanam18,ChenJinWilliams19,OliveiraPichSanthanam19,MurrayWilliams17,McKayMurrayWilliams19}.

Hardness magnification results are shown in a series of recent
papers about
\MCSP~\cite{OliveiraSanthanam18,ChenJinWilliams19,McKayMurrayWilliams19,OliveiraPichSanthanam19}.
Oliveira and Santhanam~\cite{OliveiraSanthanam18} show that
$n^{1+\epsilon}$-size lower bounds for approximating
$\MKtP[n^{\beta}]$ with an additive error $O(\log n)$ implies $\EXP\not\subseteq
\ppoly$. Oliveira, Pich and Santhanam~\cite{OliveiraPichSanthanam19}
show that for all small $\beta>0$, $n^{1+\epsilon}$-size lower
bounds for approximating $\MCSP[n^{\beta m}]$ with factor $O(m)$
error implies $\NP\not\subseteq \ppoly$. McKay, Murray, and
Williams~\cite{McKayMurrayWilliams19} show that an
$\Omega(n\poly(\log n))$ lower bound on $\poly(\log n)$ space
deterministic streaming model for $\MCSP[\poly(\log n)]$ implies
separation of P from NP.


The hardness magnification of non-uniform complexity for  sparse
sets is recently developed by Chen, Jin and
Williams~\cite{ChenJinWilliams19}. Since $\MCSP[s(n)]$ are of
sub-exponential density for $s(n)=n^{o(1)}$, the hardness
magnification  for sub-exponential density sets is more general than
the hardness magnification  for $\MCSP$. They show that if there is
an $\epsilon>0$ and a family of languages  $\{L_b\}$ (indexed over
$b\in (0,1)$) such that each $L_b$ is a $2^{n^b}$-sparse language in
$\NP$, and $L_b\not\in \Circuit[n^{1+\epsilon}]$, then
$\NP\not\subseteq \Circuit[n^k]$ for all $k$, where $\Circuit[f(n)]$
is the class of languages with nonuniform circuits of size bounded
by function $f(n)$. Their result also holds for all complexity
classes $C$ with $\exists C=C$.

On the other hand, it is unknown if \MCSP~is NP-hard. Murray and
Williams ~\cite{MurrayWilliams17} show that \NP-completeness of
\MCSP~implies the separation of \EXP~ from ZPP, a long standing
unsolved problem in computational complexity theory. Hitchcock and
Pavan~\cite{HitchcockPavan15,MurrayWilliams17} if \MCSP~is
$\NP$-hard under polynomial time truth-table reductions, then
EXP$\not\subseteq \NP \cap \ppoly$.

Separating \NEXP~from \BPP, and \EXP~from \ZPP~ are two of major
open problems in the computational complexity theory. We are
motivated by further relationship about sparse sets and \MCSP, and
the two separations $\NEXP\not=\BPP$ and $\EXP\not=\ZPP$. We develop
a polynomial method on finite fields to magnify  the hardness of
sparse sets in nondeterministic time complexity classes over a
randomized streaming model. One of our results show that if there
exists a $2^{n^{o(1)}}$-sparse set in $\NTIME(2^{n^{o(1)}})$ that
does not have a randomized streaming algorithm with $n^{o(1)}$ updating
time, and $n^{o(1)}$ space, then $\NEXP\not=\BPP$, where a
$f(n)$-sparse set is a language that has at most $f(n)$ strings of
length $n$. Our magnification result has a flexible trade off
between the spareness and time complexity.

We use two functions $d(n)$ and $g(n)$ to control the sparseness of
a tally set $T$. Function $d(n)$ gives an upper bound for the number
of elements of in $T$ and $g(n)$ is the gap lower bound between a
string $1^n$ and the next string $1^m$ in $T$, which satisfy
$g(n)<m$. The class $\tally(d(n), g(n))$ defines the class of all
those tally sets. By choosing $d(n)=\log\log n$, and
$g(n)=2^{2^{2n}}$, we prove that if~\MCSP~is $\ZPP\cap \tally(d(n),
g(n))$-hard under polynomial time truth-table reductions, then
$\EXP\not=\ZPP$.

\subsection{Comparison with the existing results}
Comparing with some existing results about sparse sets hardness
magnification in this line~\cite{ChenJinWilliams19}, there are some
new advancements in this paper.
\begin{enumerate}[1.]
    \item
    Our magnification  of sparse set is based on a
    uniform streaming model. A class of results
    in~\cite{ChenJinWilliams19} are based on nonuniform models.
    In~\cite{McKayMurrayWilliams19}, they show that if there is
    $A\in\PH$, and a function $s(n)\ge \log n$,  search-$\MCSP^A[s(n)]$
    does not have $s(n)^c$ updating time in deterministic streaming
    model for all positive, then ${\rm P}\not=\NP$. $\MCSP[s(n)]$ is a
    $s(n)^{\bigO(s(n))}$-sparse set.
    \item
    Our method is conceptually simple, and easy to understand. It is a
    polynomial algebraic approach on finite fields.
    \item
    A flexible trade off between sparseness and time complexity is given
    in our paper.
\end{enumerate}

Proving NP-hardness for \MCSP~implies
$\EXP\not=\ZPP$~\cite{HitchcockPavan15,MurrayWilliams17}.  We
consider the implication of \ZPP-hardness for $\MCSP$, and show that
if \MCSP~is $\ZPP\cap \tally(d(n), g(n))$-hard for a function pair
such as $d(n)=\log\log n$ and $g(n)=2^{2^{2n}}$, then
$\EXP\not=\ZPP$. It seems that proving \MCSP~is $\ZPP$-hard is much
easier than proving \MCSP~is \NP-hard since $\ZPP\subseteq
(\NP\cap\coNP)\subseteq \NP$. According to the low-high hierarchy theory
developed by Sch\"{o}ning~\cite{SchoningHigh}, the class
$\NP\cap\coNP$ is the low class $L_1$. Although $\MCSP$ may not be
in the class $\ZPP$, it is possible to be $\ZPP$-hard.

\section{Notations}\label{notation-sec}

Minimum Circuit Size Problem (\MCSP) is that given an integer $k$,
and a binary string $T$ of length $n=2^m$ for some integer $m\ge 0$,
determine if $T$ can be generated by a circuit of size $k$. Let
$\mathbb{N}=\{1,2,\cdots\}$ be the set of all natural numbers. For a
language $L$, $L^n$ is the set of strings in $L$ of length $n$, and
$L^{\le n}$ is the set of strings in $L$ of length at most $n$. For
a finite set $A$, denote $|A|$ to be the number of elements in $A$.
For a string $s$, denote $|s|$ to be its length.
 If $x,y,z$ are not empty strings, we
have a coding method that converts a $x, y$ into a string $\langle
x, y\rangle$ with $|x|+|y|\le |\langle x, y\rangle|\le 3(|x|+|y|)$
and converts $x,y,z$ into $\langle x, y,z\rangle$ with
$|x|+|y|+|z|\le |\langle x, y, z\rangle|\le 3(|x|+|y|+|z|)$. For
example, for $x=x_1\cdots x_{n_1}, y=y_1\cdots y_{n_2}, z=z_1\cdots
z_{n_3}$, let $\langle x, y, z\rangle=1x_1\cdots
1x_{n_1}001y_1\cdots 1y_{n_2}001z_1\cdots 1z_{n_3}$.

Let $\DTIME(t(n))$ be the class of languages accepted by
deterministic Turing machines in time $\bigO(t(n))$. Let
$\NTIME(t(n))$ be the class of languages accepted by
nondeterministic Turing machines in time $\bigO(t(n))$. Define
$\EXP=\cup_{c=1}^{\infty}\DTIME(2^{n^c})$ and
$\NEXP=\cup_{c=1}^{\infty}\NTIME(2^{n^c})$. P/poly, which is also
called PSIZE, is the class of languages that have polynomial-size
circuits.

We use a polynomial method on a finite field $F$. It is classical
theory that each finite field is of size $p^k$ for some prime number
$p$ and integer $k\ge 1$ (see~\cite{Hungerford74}). For a finite
field $F$, we denote $R(F)=(p, t_F(u))$ to represent $F$, where
$t_F(u)$ is a irreducible polynomial over field $\GF(p)$ for the
prime number $p$ and its degree is $\deg(t_F(.))=k$. The polynomial
$t_F(u)$ is equal to $u$ if $F$ is of size $p$, which is a prime
number. Each element of $F$ with $R(F)=(p, t_F(u))$ is a polynomial
$q(u)$ with degree less than the degree of $t_F(u)$. For two
elements $q_1(u)$ and $q_2(u)$ in $F$, their addition is defined by
$(q_1(u)+q_2(u))(\mod\ t_F(u))$,  and their multiplication is
defined by $(q_1(u)\cdot q_2(u))(\mod\ t_F(u))$
(see~\cite{Hungerford74}). Each element in $\GF(2^k)$ is a
polynomial $\sum_{i=0}^{k-1}b_iu^i$ ($b_i\in\{0,1\}$), which is
represented by a binary string $b_{k-1}\cdots b_0$ of length $k$.

We use $\GF(2^k)$ field in our randomized streaming algorithm for
hardness magnification . Let $F$ be a $\GF(2^k)$ field (a field of
size $q=2^k$) and has its $R(F)=(2, t_F(u))$. Let $s=a_0\cdots
a_{m-1}$ be a binary string of length $m$ with $m\le k$, and $u$ be
a variable. Define $w(s, u)$ to be the element $\sum_{i=0}^{m-1} a_i
u^i$ in $\GF(2^k)$.
Let $x$ be a string in $\{0,1\}^*$ and $k$ be an integer at least
$1$. Let $x=s_{r-1}s_{t-2}\cdots s_1s_0$ such that each $s_i$ is a
substring of $x$ of length $k$ for $i=1,2,\cdots, r-1$, and the
substring $s_0$ has its length $|s_0|\le k$. Each $s_i$ is called a
{\it $k$-segment} of $x$ for $i=0,1,\cdots, r-1$. Define the
polynomial $d_x(z)=z^r+\sum_{i=0}^{r-1}w(s_i, u) z^i$, which
converts a binary string into a polynomial in $\GF(2^k)$.

We develop a streaming algorithm that converts an input string into
an element in a finite field. We give the definition to characterize
the properties of the streaming algorithm developed in this paper.
Our streaming algorithm is to convert an input stream $x$ into an
element $d_x(a)\in F=\GF(2^k)$ by selecting a random element $a$
from $F$.

\begin{definition}
    Let $r_0(n), r_1(n), r_2(n), s(n), u(n)$ be nondecreasing functions from
    $\mathbb{N}$ to $\mathbb{N}$. Define $\streaming(r_0(n), r_1(n),  s(n),  u(n), r_2(n))$
    to be the class of languages $L$ that have one-pass streaming algorithms that
    has input $(n, x)$ with $n=|x|$ ($x$ is a string and read by
    streaming), it satisfies
    \begin{enumerate}
        \item
        It takes $r_0(n)$ time to generate a field $F=\GF(2^k)$, which is
        represented by $(2, t_F(.))$ with a irreducible polynomial $t_f(.)$ over
        $\GF(2)$ of degree $k$.
        \item
        It takes $O(r_1(n))$  random steps before reading the first bit from
        the input stream $x$.
        \item
        It uses $O(s(n))$ space that includes the space to hold the field
        representation generated by the algorithm. The space for a field
        representation is $\Omega( (\deg(t_F(.))+1))$ and $\bigO(
        (\deg(t_F(.))+1))$ for the irreducible polynomial $t_F(.)$ over
        $\GF(2)$.

        \item
        It takes $O(u(n))$ field conversions to elements in $F$ and $O(u(n))$
        field operations in $F$ after reading each bit.
        \item
        It runs $O(r_2(n))$ randomized steps after reading the entire input.

    \end{enumerate}
\end{definition}

\section{Overview of Our Methods}
In this section, we give a brief description about our methods used
in this paper. Our first result is based on a polynomial
method on a finite field whose size affects the hardness of
magnification. The second result is a translational method for
zero-error probabilistic complexity classes.

\subsection{Magnify  the Hardness of Sparse Sets}
We have a polynomial method over finite fields. Let $L$ be
$f(n)$-sparse language in $\NTIME(t_1(n))$. In order to handle an
input string of size $n$, a finite field $F=\GF(q)$ with $q=2^k$ for
some integer $k$ is selected, and is represented by $R(F)=(2,
t_F(z))$, where $t_F(z)$ is a irreducible polynomial over $\GF(2)$.
An input $y=a_1a_2\cdots a_n$ is partitioned into $k$-segments
$s_{r-1}\cdots s_1s_0$ such that each $s_i$ is converted into an
element $w(s_i, u)$ in $F$, and $y$ is transformed into an
polynomial $d_y(z)=z^r+\sum_{i=0}^{r-1}w(s_i, u) z^i$. A random
element $a\in F$ is chosen in the beginning of streaming algorithm
before processing the input stream. The value $d_y(a)$ is evaluated
with the procession of input stream. The finite $F$ is large enough
such that for different $y_1$ and $y_2$ of the same length,
$d_{y_1}(.)$ and $d_{y_2}(.)$ are different polynomials due to their
different coefficients derived from $y_1$ and $y_2$, respectively.
Let $H(y)$ be the set of all $\langle n, a, d_y(a)\rangle$ with
$a\in F$ and $n=|y|$. Set $A(n)$ is the union of all $H(y)$ with
$y\in L^n$. The set of $A$ is $\cup_{i=1}^{\infty}A(n)$. A small
lower bound for the language  $A$ is magnified to large lower bound
for $L$.

The size of field $F$ depends on the density of set $L$ and is
$\bigO(f(n)n)$. By the construction of $A$, if $y\in L$, there are
$q$ tuples $\langle n, a, d_y(a)\rangle$ in $A$ that are generated
by $y$ via all $a$ in $F$. For two different $y_1$ and $y_2$ of
length $n$, the intersection $H(y_1)\cap H(y_2)$ is bounded by the
degree of $d_{y_1}(.)$. If $y\not\in L$, the number of items
$\langle n, a, d_y(a)\rangle$ generated by $y$ is at most ${q\over
    4}$ in $A$. If $y\in L$, the number of items $\langle n, a,
d_y(a)\rangle$ generated by $y$ is $q$ in $A$.  This enables us to
convert a string $x$ of length $n$ in $L$ into some strings in $A$
of length much smaller than $n$, make the hardness magnification
possible.

\subsection{Separation by \ZPP-hardness of \MCSP}

Our another result shows that \ZPP-hardness for \MCSP~implies
$\EXP\not=\ZPP$. We identify a class of functions that are padding
stable, which has the property if $T\in \tally(d(n), g(n))$, then
$\{1^{n+2^n}: 1^n\in T\}\in \tally(d(n), g(n))$. The function pair
$d(n)=\log\log n$ and $g(n)=2^{2^{2n}}$ has this property.
We construct a very sparse tally set $L\in \EXP\cap
\tally(d(n), g(n))$ that separates $\ZEP$ from $\ZPP$, where $\ZEP$
is the zero error exponential time probabilistic  class. It is based
on a diagonal method that is combined with a padding design. A tally
language $L$ has a zero-error $2^{2^n}$-time probabilistic algorithm
implies $L'=\{1^{n+2^n}: 1^n\in L\}$ has a zero-error $2^n$-time
probabilistic algorithm. Adapting to the method
of~\cite{MurrayWilliams17}, we prove that if \MCSP~is $\ZPP\cap
\tally(d(n), g(n))$-hard under polynomial time truth-table
reductions, then $\EXP\not=\ZPP$.

\section{Hardness Magnification via Streaming }
In this section, we show a hardness magnification of sparse sets via
a streaming algorithm. A classical algorithm to find irreducible
polynomial~\cite{shoup88} is used to construct a field that is large
enough for our algorithm.

\begin{theorem}\cite{shoup88}\label{det-irreducible-thm}
    There is a deterministic algorithm that constructs a irreducible
    polynomial of degree $n$ in $\bigO(p^{1\over 2}(\log p)^3
    n^{3+\epsilon}+(\log p)^2 n^{4+\epsilon})$ operations in $F$, where
    $F$ is a finite field $\GF(p)$ with prime number $p$.
\end{theorem}

\begin{definition}
    Let $f(n)$ be a function from $N$ to $N$. For a language $A\subseteq
    \{0,1\}^*$, we say $A$ is $f(n)$-sparse if $|A^n|\le f(n)$ for all
    large integer $n$.
\end{definition}

\subsection{Streaming Algorithm}


The algorithm Streaming(.) is based on a language $L$ that is
$f(n)$-sparse. It generates a field $F=\GF(2^k)$ and evaluates
$d_x(a)$ with a random element $a$ in $F$. A polynomial
$z^r+\sum_{i=0}^{r-1}
b_iz^i=z^r+b_{r-1}z^{r-1}+b_{r-2}z^{r-2}+\cdots+b_0$ can be
evaluated by $(\cdots ((z+b_{r-1})z+b_{r-2})z+...)z+b_0$ according
to the classical Horner's algorithm. For example,
$z^2+z+1=(z+1)z+1$.


{\bf Algorithm}

Streaming($n, x$)

Input: an integer $n$, and string $x=a_1\cdots a_n$ of the length
$n$;

Steps:

\begin{enumerate}[1.]

    \item\label{selecting-k-line}
    Select a field size $q=2^k$ such that $8f(n)n<q\le 16f(n)n$.


    \item\label{generate_field_line}
    Generate an irreducible polynomial $t_F(u)$ of degree $k$ over
    $\GF(2)$ such that $(2, t_F(u))$ represents finite $F=\GF(q)$ (by
    Theorem~\ref{det-irreducible-thm} with $p=2$);

    \item
    Let $a$ be a random element in $F$;

    \item
    Let $r=\ceiling{n\over k}$; (Note that $r$ is the number of
    $k$-segments of $x$. See Section~\ref{notation-sec})

    \item
    Let $j=r-1$;

    \item
    Let $v=1$;

    \item
    Repeat

    \item
    \{

    \item
    \hskip 20pt Receive the next $k$-segment $s_j$ from  the input stream
    $x$;

    \item\label{convert-to-field-line}
    \hskip 20pt Convert $s_j$ into an element $b_j=w(s_j, u)$ in $\GF(q)$;

    \item\label{next-substring-line}
    \hskip 20pt Let $v=v\cdot a +b_j$;

    \item
    \hskip 20pt Let $j=j-1$;

    \}

    \item
    Until $j<0$ (the end of the stream);




    \item
    Output $\langle n, a, v\rangle$;

\end{enumerate}

{\bf End of Algorithm}


Now we have our magnification  algorithm. Let $M(.)$ be a randomized
Turing machine to accept a language $A$ that contains all $\langle
|x|, a, d_x(a)\rangle$ with $a\in F$ and $x\in L$. We have the
following randomized streaming algorithm to accept $L$ via the
randomized algorithm $M(.)$  for $A$.


{\bf Algorithm}

Magnification($n,x$)

Input integer $n$ and $x=a_1\cdots a_n$ as a stream;

Steps:
    Let $y=$Streaming($n,x$);     Accept if $M(y)$ accepts;


{\bf End of Algorithm}


\subsection{Hardness Magnification }\label{magnification -sec}

In this section, we derive some results about hardness magnification
via sparse set. Our results show a trade off between the hardness
magnification  and sparseness via the streaming model.

\begin{definition}
        For a nondecreasing function $t(.):\mathbb{N}\rightarrow \mathbb{N}$, define
        $\PrTime(t(n))$ the class of languages $L$ that have
        two-side  bounded error probabilistic algorithms with time complexity
        $O(t(n))$.
      Define  $\BPP=\cup_{c=1}^{\infty}\PrTime(n^c)$.
\end{definition}

\begin{theorem}\label{main-thm}Assume that $u_1(m)$ be nondecreasing function for the time to
    generate an irreducible polynomial of degree $m$ in $\GF(2)$, and
    $u_2(m)$ be the nondecreasing function of a time upper bound for
    the operations ($+, .$) in
    $\GF(2^m)$.
    Let $f(.), t_1(.), t_2(.), t_3(n)$ be nondecreasing functions $\mathbb{N}\rightarrow \mathbb{N}$ with
    $f(n)\le 2^{n\over 2}$, $v(n)=(\log n+\log f(n))$, and $10v(n)+t_1(n)+u_1(10v(n))+n\cdot
    u_2(10v(n)) \le t_2(v(n))$
    for all large $n$.
    If there is a $f(n)$-sparse set $L$ with $L\in \NTIME(t_1(n))$ and
    $L\not\in \streaming(u_1(10v(n))), v(n), v(n), 1,t_3(10 v(n)))$,
    then there is a language $A$ such that $A\in \NTIME(t_2(n))$ and
    $A\not\in\PrTime(t_3(n))$.
\end{theorem}

\begin{proof}
    Select a finite field $\GF(q)$ with $q=2^k$ for  an integer $k$ by line~\ref{selecting-k-line} of the algorithm streaming(.).
    For each $x\in L^n$,
    let $x$ be partitioned
    into $k$-segments: $s_{r-1}s_{r-2}\cdots s_0$.  Let $w(s_i,u)$
    convert $s_i$ into  an element of $\GF(q)$ (See
    Section~\ref{notation-sec}). Define polynomial
    $d_x(z)=z^r+\sum_{i=0}^{r-1}w(q,s_i) z^i$. For each $x$, let $H(x)$ be
    the set $\{\langle n, a, d_x(a)\rangle |a\in \GF(q)\}$, where
    $n=|x|$. Define set $A(n)=\cup_{y\in L^n}H(y)$ for $n=1,2\cdots$,
    and language $A=\cup_{n=1}^{+\infty}A(n)$.


    {\bf Claim 1.} For any $x\not\in L^n$ with $n=|x|$, we have
    $|H(x)\cap A(n)|< {q\over 4}$.

    \begin{proof}
        Assume that for some $x\not\in L^n$ with $n=|x|$, $|H(x)\cap
        A(n)|\ge {q\over 4}$. It is easy to see that $r\le n$ and $k\le n$
        for all large $n$ by the algorithm Streaming(.) and the condition of
        $f(.)$ in the theorem. Assume that $|H(x)\cap H(y)|<r+1$ for every
        $y\in L^n$. Since $A(n)$ is the union  $H(y)$ with $y\in L^n$ and
        $|L^n|\le f(n)$, there are at most $rf(n)\le nf(n)<{q\over 8}$
        elements in $H(x)\cap A(n)$ by line~\ref{selecting-k-line} of the
        algorithm Streaming(.). Thus, $|H(x)\cap A(n)|< {q\over 8}$. This
        brings a contradiction. Therefore, there is a $y\in L^n$ to have
        $|H(x)\cap H(y)|\ge r+1$. Since the polynomials $d_x(.)$ and
        $d_y(.)$ are of degrees at most $r$, we have $d_x(z)=d_y(z)$ (two polynomials are equal). Thus,
        $x=y$. This brings a contradiction because $x\not\in L^n$ and $y\in
        L^n$.
    \end{proof}


    {\bf Claim 2.} If  $x\in L$, then Streaming$(|x|,x)\in A$.
    Otherwise, with probability at most ${1\over 4}$, Streaming$(|x|,
    x)\in A$.

    \begin{proof}
        For each $x$, it generates $\langle  n,  a, d_x(a)\rangle $ for a
        random $a\in \GF(q)$.  Each $a\in \GF(q)$ determines a random path. We
        have that if $x\in L$, then $\langle  n,  a, d_x(a)\rangle \in A$,
        and if $x\not\in L$, then $\langle  n,  a, d_x(a)\rangle \in A$ with
        probability at most ${1\over 4}$ by Claim 1.
    \end{proof}


    {\bf Claim 3.} $A\in \NTIME(t_2(m))$.

    \begin{proof}
        Let $z=10v(n)=10(\log n+\log f(n))$. Each element in field $F=\GF(2^k)$ is of length $k$. For each $u=\langle n, a,b\rangle$ ($a,b\in F$), we
        need to guess a string $x\in L^n$ such that $b=d_x(a)$. It is easy
        to see that $v(n)\le |\langle n, a, b\rangle |\le 10v(n)$ for all large $n$ if
        $\langle n, a, b\rangle \in A$ (See Section~\ref{notation-sec} about coding). Let $m=|\langle n, a, b\rangle |$.
        It takes at most $u_1(z)$ steps to generate a irreducible polynomial
        $t_F(.)$ for the field $F$ by our assumption.

        Since $L\in \NTIME(t_1(n))$, checking if $u\in A$ takes
        nondeterministic $t_1(n)$ steps to guess a string $x\in L^n$, $u_1(z)$ deterministic steps to generate $t_F(u)$ for the field $F$,
        $O(z)$ nondeterministic steps to generate a random element $a\in F$, and
        additional $\bigO(n\cdot u_2(z))$ steps to evaluate $d_x(a)$ in by following algorithm
        Streaming(.) and check $b=d_x(a)$. The polynomial $t_F(u)$ in the
        $\GF(2)$ has degree at most $z$. Each polynomial
        operation ($+$ or $.$) in $F$ takes at most $u_2(z)$ steps. Since $z+t_1(n)+u_1(z)+n\cdot u_2(z)\le t_2(m)$ time by the condition of this theorem, we have
        $A\in\NTIME(t_2(m))$.


    \end{proof}

    {\bf Claim 4.} If $A\in\PrTime(t_3(m))$, then $L\in \streaming(u_1(10v(n)),v(n), v(n),$ $1,  t_3(10v(n)))$.

    \begin{proof}
        The field generated at line~\ref{generate_field_line} in algorithm
        Streaming(.) takes $u_1(10(\log n+\log f(n)))$ time. 
        Let $x=a_1\cdots a_n$ be the input string. The string $x$
        partitioned into $k$-segments $s_{r-1}\cdots s_0$. Transform each
        $s_i$ into an element $b_i=w(s_i, u)$ in $\GF(q)$ in the streaming
        algorithm. We generate a polynomial $d_x(z)=z^r+\sum_{i=0}^{r-1}
        b_iz^i=z^r+b_{r-1}z^{r-1}+b_{r-2}z^{r-2}+\cdots+b_0$. Given a random
        element $a\in \GF(q)$, we evaluate $d_x(a)=(\cdots
        ((a+b_{r-1})a+b_{r-2})a+...)a+b_0$ according to the classical
        algorithm. Therefore, $d_x(a)$ is evaluated in Streaming(.) with
        input $(|x|,x)$.

        If $A\in\PrTime(t_3(m))$, then $L$ has a randomized streaming
        algorithm that has  at most $t_3(10v(n))$ random steps
        after reading the input, and at most $\bigO(v(n))$ space.
        After reading one substring $s_i$ from $x$, it takes one conversion
        from a substring of the input to an element of field $F$ by line~\ref{convert-to-field-line}, and at most two field
        operations by line~\ref{next-substring-line} in the algorithm
        Streaming(.).

    \end{proof}

    Claim 4 brings a contradiction to our assumption about the
    complexity of $L$ in the theorem. This proves the theorem.
\end{proof}

\begin{proposition}\label{o-ineqn-proposition} Let $f(n):\mathbb{N}\rightarrow \mathbb{N}$ be a nondecreasing
    function. If for each fixed $\epsilon\in (0,1)$, $f(n)\le
    n^{\epsilon}$ for all large $n$, then there is a nondecreasing
    unbounded function $g(n):\mathbb{N}\rightarrow \mathbb{N}$ with $f(n)\le n^{1\over
        g(n)}$.
\end{proposition}

\begin{proof}Let $n_0=1$. For each $k\ge 1$,
    let $n_k$ be the least integer such that $n_k\ge n_{k-1}$ and $f(n)\le
    n^{1\over k}$ for all $n\ge n_k$. Clearly, we have the infinite list
    $n_1\le n_2\cdots\le n_k\le\cdots$ such that $\lim_{k\rightarrow
        +\infty}n_k=+\infty$.  Define function $g(k):\mathbb{N}\rightarrow \mathbb{N}$ such
    that $g(n)=k$ for all $n\in [n_{k-1}, n_k)$. For each $n\ge n_k$, we
    have $f(n)\le n^{1\over k}$.
\end{proof}

Our Definition~\ref{o-ineqn-def} is based
Proposition~\ref{o-ineqn-proposition}.  It can simplify the proof
when we handle a function that is $n^{o(1)}$.

\begin{definition}\label{o-ineqn-def}
        A function $f(n):\mathbb{N}\rightarrow \mathbb{N}$ is $n^{o(1)}$ if there is a
        nondecreasing function $g(n):\mathbb{N}\rightarrow \mathbb{N}$ such that
        $\lim_{n\rightarrow +\infty}g(n)=+\infty$ and $f(n)\le n^{1\over
            g(n)}$ for all large $n$.
        A function $f(n):\mathbb{N}\rightarrow \mathbb{N}$ is $2^{n^{o(1)}}$ if there is a
        nondecreasing function $g(n):\mathbb{N}\rightarrow\mathbb{N}$ such that
        $\lim_{n\rightarrow +\infty}g(n)=+\infty$ and $f(n)\le 2^{n^{1\over
                g(n)}}$ for all large $n$.
\end{definition}


\begin{corollary}
    If there exists a $2^{n^{o(1)}}$-sparse language $L$ in
    $\NTIME(2^{n^{o(1)}})$ such that $L$ does not have any randomized
    streaming algorithm with $n^{o(1)}$ updating time, and $n^{o(1)}$
    space, then $\NEXP\not=\BPP$.
\end{corollary}
\begin{proof}
    Let $g(n):\mathbb{N}\rightarrow \mathbb{N}$ be an arbitrary unbounded nondecreasing
    function that satisfies $\lim_{n\rightarrow+\infty}g(n)=+\infty$ and
    $g(n)\le \log\log n$. Let $t_1(n)=f(n)=2^{n^{1\over g(n)}}$ and Let
    $t_2(n)=2^{2n}$, $t_3(n)=n^{\sqrt{g(n)}}$, and $v(n)=(\log
    n+\log f(n))$.

    It is easy to see that $v(n)=n^{\littleo(1)}$, and both $u_1(n)$
    and $u_2(n)$ are $n^{\bigO(1)}$ (see Theorem~\ref{det-irreducible-thm}). For any fixed $c_0>0$, we have
    $t_2(v(n))>t_2(\log f(n))\ge t_2(n^{1\over
        g(n)})>t_1(n)+n^{c_0}$ for all large $n$.  For all large $n$, we have
    \begin{eqnarray}
    t_3(10v(n))&\le& t_3(20\log f(n))=t_3(20n^{1\over g(n)})\\
    &\le& (20n^{1\over g(n)})^{\sqrt{g(20n^{1\over g(n)}})}   \le (n^{2\over g(n)})^{\sqrt{g(n)}}=n^{o(1)}.\label{t3-f-ineqn}
    \end{eqnarray}
Clearly, these functions
   satisfy the inequality of the precondition in
    Theorem~~\ref{main-thm}.
    Assume $L\in \streaming(\poly(v(n)), v(n), v(n), 1, t_3(10v(n)))$.
    With $O(v(n))=n^{o(1)}$ space, we have a field representation $(2, t_F(.))$
    with $\deg(t_F(.))=n^{o(1)}$. Thus, each field operation takes
    $n^{o(1)}$ time by the brute force method for polynomial
    addition and multiplication. We have $t_3(10v(n))=n^{o(1)}$ by inequality (\ref{t3-f-ineqn}). Thus, the
    streaming algorithm updating time is $n^{o(1)}$. Therefore,
    we have that $L$ has a randomized streaming algorithm with
    $n^{o(1)}$ updating time, and $n^{o(1)}$ space. This gives a
    contradiction. So, \\$L\not\in \streaming(\poly(v(n)),
    v(n), v(n), 1,t_3(10v(n)))$. By
    Theorem~\ref{main-thm}, there is $A\in \NTIME(t_2(n))$ such that
    $A\not\in\PrTime(t_3(n))$. Therefore, $A\not\in \BPP$. Thus,
    $\NEXP\not=\BPP$.
\end{proof}

\section{Implication of \ZPP-Hardness of~\MCSP}

In this section, we show that if \MCSP~is $\ZPP\cap\tally$-hard,
then $\EXP\not=\ZPP$. The conclusion still holds if $\tally$~is
replaced by a very sparse subclass of $\tally$ languages.


\begin{definition}
        For a nondecreasing function $t(.):\mathbb{N}\rightarrow \mathbb{N}$, define
        $\ZPrTime(t(n))$ the class of languages $L$ that have   zero-error
        probabilistic algorithms with time complexity $O(t(n))$.
        Define $\ZPP=\cup_{c=1}^{\infty}\ZPrTime(n^c)$, and\\    $\ZEP=\cup_{c=1}^{\infty}\ZPrTime(2^{n^c})$.
\end{definition}

\begin{definition}
        For an nondecreasing function $f(n):\mathbb{N}\rightarrow \mathbb{N}$, define $\tally[f(k)]$ to be the
        class of tally set $A\subseteq \{1\}^*$ such that for each $1^m\in
        A$, there is an integer $i\in \mathbb{N}$ with $m=f(i)$.
        For a tally language $T\subseteq \{1\}^*$, define
        $\Pad(T)=\{1^{2^n+n}| 1^n\in T\}$.
\end{definition}

\begin{definition}
    For two languages $A$ and $B$, a polynomial time {\it truth-table
        reduction} from $A$ to $B$ is a polynomial time computable function
    $f(.)$ such that for each instance $x$ for $A$, $f(x)=(y_1,\cdots,
    y_m, C(.))$ to satisfy $x\in A$ if and only if $C(B(y_1),\cdots,
    B(y_m))=1$, where $C(.)$ is circuit of $m$ input bits and $B(.)$ is
    the characteristic function of $B$.

    Let $\le_r^P$ be a type of polynomial time reductions ($\le^P_{tt}$
    represents polynomial time truth-table reductions), and $C$ be a
    class of languages. A language $A$ is $C$-hard under $\le_r^P$
    reductions if for each $B\in C$, $B\le_r^P A$.
\end{definition}

\begin{definition}
    Let $k$ be an integer. Define two classes of functions with
    recursions:
      (1)  $\log^{(1)}(n)=\log_2 n$, and
        $\log^{(k+1)}(n)=\log_2(\log^{(k)}(n))$.
        (2) $\ex^{(1)}(n)=2^n$, and $\ex^{(k+1)}(n)=2^{\ex^{(k)}(n)}$.
\end{definition}

\begin{definition}
    For two nondecreasing functions $d(n), g(n):\mathbb{N}\rightarrow \mathbb{N}$, the
    pair $(d(n), g(n))$ is {\it time constructible} if $(d(n), g(n))$
    can be computed in time $d(n)+g(n)$ steps.
\end{definition}

\begin{definition}
    Define $\tally(d(n), g(n))$ to be the class of tally sets $T$ such
    that $|T^{\le n}|\le d(n)$ and for any two strings $1^n, 1^m\in T$
    with $n<m$, they satisfy $g(n)<m$. We call $d(n)$ to be the {\it
        density function} and $g(n)$ to be the {\it gap function}. A gap
    function $g(n)$ is {\it padding stable} if $g(2^n+n)<2^{g(n)}+g(n)$
    for all $n>1$.
\end{definition}

\begin{lemma}\label{padding-lemma}\scrod
    \begin{enumerate}
        \item\label{part1}
        Assume the gap function $g(n)$ is padding stable. If $T\in
        \tally(d(n), g(n))$, then $\Pad(T)\in \tally(d(n), g(n))$.
        \item\label{part2}
        For each integer $k>0$, $g(n)=\ex^{(k)}({2n})$ is padding stable.
    \end{enumerate}
\end{lemma}

\begin{proof}
    Part~\ref{part1}. Let $1^n$ be a string in $T$. The next shortest string $1^m\in T$
    with $n<m$ satisfies $g(n)<m$. We have $1^{2^n+n}$ and $1^{2^m+m}$
    are two consecutive neighbor strings in $\Pad(T)$ such that there is
    no other string $1^k\in \Pad(T)$ with $2^n+n<k<2^m+m$. We have
    $g(2^n+n)<2^{g(n)}+g(n)<2^m+m$. Since the strings in $\Pad(T)^{\le
        n}$ are one-one mapped from the strings in $T$ with length less than
    $n$, $|\Pad(T)^{\le n}|\le |T^{\le n}|\le d(n)$, we have $\Pad(T)\in
    \tally(d(n), g(n))$. This proves Part~(\ref{part1}).

    Part~\ref{part2}. We have inequality $g(2^n+n)=\ex^{(k)}(2(2^n+n))<\ex^{(k)}(4\cdot
    2^n)=\ex^{(k)}(2^{n+2})$ $\le
    \ex^{(k)}(2^{2n})=2^{g(n)}<2^{g(n)}+g(n)$. Therefore, gap function
    $g(n)$ is padding stable. This proves Part~\ref{part2}.
\end{proof}

\begin{lemma}\label{f()-lemma}
    Let $d(n)$ and $g(n)$ be nondecreasing unbounded functions from $N$
    to $N$, and $(d(n), g(n))$ is time constructible. Then there exists
    a time constructible increasing unbounded function
    $f(n):\mathbb{N}\rightarrow \mathbb{N}$ such that \\$\tally[f(n)]\subseteq \tally(d(n),
    g(n))$.
\end{lemma}

\begin{proof}Compute the least integer $n_1$ with $d(n_1)>0$. Let $s_1$ be the number of steps for the computation. Define $f(1)=\max(s_1, n_1)$.
    Assume that $f(k-1)$ has been defined. We determine the function
    value $f(k)$ below.

    For an integer $k>0$, compute $g(f(k-1))$ and the least $k$ numbers $n_1<n_2<\cdots<n_k$ such that
    $0<d(n_1)<d(n_2)<\cdots<d(n_k)$. Assume the computation above takes
    $s$ steps. Define $f(k)$ to be the $\max(2s, n_k, g(f(k-1))+1)$. For
    each language $T\in \tally[f(n)]$, there are at most $k$ strings in
    $T$ with length at most $f(k)$. On the other hand, $d(n_k)\ge k$ by
    the increasing list $0<d(n_1)<d(n_2)<\cdots<d(n_k)$. Therefore, we
    have $|T^{\le n_k}|\le k\le d(n_k)$. Furthermore, we also have
    $g(f(k-1))<f(k)$. Since $s$ is the number of steps to determine the
    values $s, n_k, $ and $g(f(k-1))+1$. We have $2s\le f(k)$.  Thus,
    $f(k)$ can be computed in $f(k)$ steps by spending some idle steps.
    Therefore, the function $f(.)$ is time constructible.
\end{proof}

We will use the notion $\tally[f(k)]$ to characterize extremely
sparse tally sets with fast growing function such as
$f(k)=2^{2^{2^k}}$. It is easy to see that $\tally=\tally[I(.)]$, where
$I(.)$ is the identity function $I(k)=k$.

\begin{lemma}\label{diag-lemma}
    Let $d(n)$ and $g(n)$ be nondecreasing unbounded functions. If
    function $g(n))$ is padding stable, then there is a language $A$
    such that $A\in\ZPrTime(2^{O(n)})\cap \tally(d(n), g(n))$ and
    $A\not\in \ZPP$.
\end{lemma}

\begin{proof}
It is based on the classical translational method. Assume $\ZPrTime(2^{O(n)})\cap\tally(d(n), g(n))\subseteq \ZPP$.  Let $f(.)$ be a time constructible
    increasing unbounded function via Lemma~\ref{f()-lemma} such that\\
    $\tally[f(n)]\subseteq \tally(d(n), g(n))$. Let $t_1(n)=2^{2^n}$ and
    $t_2(n)=2^{2^{n-1}}$.  Let $L$ be a tally language in
    $\DTIME(t_1(n))\cap \tally[f(n)]$, but it is not in
    $\DTIME(t_2(n))$. Such a language $L$ can be constructed via a
    standard diagonal method. Let $M_1,\cdots, M_2$ be the list of
    Turing machines such that each $M_i$ has time upper bound by
    function $t_2(n)$. Define language $L\in \tally[f(n)]$ such that for
    each $k$, $1^{f(k)}\in L$ if and only if $M_k(1^{f(k)})$ rejects in
    $t_2(f(k))$ steps. We have $L\in\tally(d(n), g(n))$ by
    Lemma~\ref{f()-lemma}.

    Let $L_1=\Pad(L)$. We have $L_1\in\tally(d(n), g(n))$ by Lemma~\ref{padding-lemma}. We
    have $L_1\in \DTIME(2^{O(n)})\subseteq \ZPrTime(2^{O(n)})$. Thus,
    $L_1\in \ZPP$. So, $L\in \ZPrTime(2^{O(n)})$. Therefore, $L\in
    \ZPrTime(2^{O(n)})\cap \tally(d(n), g(n))$. We have $L\in \ZPP$.
    Thus, $L\in \DTIME(2^{n^{\bigO(1)}})\subseteq \DTIME(2^{2^{n-1}})$.
    This brings a contradiction.
\end{proof}

\begin{theorem}\label{second-main-thm} Let $d(n)$ and $g(n)$ be nondecreasing unbounded functions from $\mathbb{N}$ to $\mathbb{N}$. Assume that $g(n)$ is padding stable.
    If \MCSP~ is $\ZPP\cap \tally(d(n), g(n))$-hard under polynomial
    time truth-table reductions, then \EXP$\not=\ZPP$.
\end{theorem}

\begin{proof}
    Assume that \MCSP~is $(\ZPP\cap\tally(d(n), g(n))$-hard under
    polynomial time truth-table reductions, and $\EXP=\ZPP$.

    Let $L$ be a language in $\ZPrTime(2^{O(n)})\cap\tally[d(.), g(.)]$,
    but $L\not\in \ZPP$ by Lemma~\ref{diag-lemma}. Let $L'=\Pad(L)$.
    Clearly, every string $1^y$ in $L'$ has the property that $y=2^n+n$
    for some integer $n$. This property is easy to check and we reject
    all strings without this property in linear time. We have $L'\in
    \ZPP$. Therefore, there is a polynomial time truth-table reduction
    from $L'$ to~\MCSP~via a polynomial time truth-table reduction
    $M(.)$. Let polynomial $p(n)=n^c$ be the running time for $M(.)$ for a
    fixed $c$ and $n\ge 2$.

    Define the language $R=\{(1^n,i,j),$ the $i$-th bit of $j$-th query
    of $M(1^{n+{2^n}})$ is equal to $1,$ and $ i,j\le p(n+2^n) \}$. We
    can easily prove that $R$ is in $\EXP$. Therefore, $R\in
    \ZPP\subseteq \ppoly$ (See~\cite{Adleman}).


    Therefore, there is a class of polynomial size circuits
    $\{C_n\}_{n=1}^{\infty}$ to recognize $R$ such that $C_n(.)$
    recognize all $(1^n,i,j)$ with $ i,j\le p(n+2^n)$ in $R$. Assume
    that the size of $C_n$ is of size at most $q(n)=n^{t_0}+t_0$ for a
    fixed $t_0$.
    For an instance $x=1^n$ for $L$, consider the instance $y=1^{n+2^n}$
    for $L'$. We can compute all non-adaptive queries $\langle T,
    s(n)\rangle$ to $\MCSP$ in $2^{n^{O(1)}}$ time via $M(y)$. If
    $s(n)\ge q(n)$, the answer from $\MCSP$~for the query $\langle T,
    s(n)\rangle$ is yes since $\langle T, s(n)\rangle$ can be generated
    as one of the instances via the circuit $C_n(.)$. If $s(n)<q(n)$, we
    can use a brute force method to check if there exists a circuit of
    size at most $q(n)$ to generate $T$. It takes $2^{n^{O(1)}}$ time.
    Therefore, $L\in\EXP$. Thus, $L\in \ZPP$. This bring a contradiction
    as we already assume $L\not\in \ZPP$.
\end{proof}

\begin{corollary}For any integer $k$, if \MCSP~ is $\ZPP\cap \tally(\log^{(k)}(n), \ex^{(k)}(2n))$-hard under
    polynomial time truth-table reductions, then \EXP$\not=\ZPP$.
\end{corollary}

\begin{proof}
    It follows from Theorem~\ref{second-main-thm} and
    Lemma~\ref{padding-lemma}.
\end{proof}

\begin{corollary}For any integer $k$, if \MCSP~ is $\ZPP\cap\tally$-hard under
    polynomial time truth-table reductions, then \EXP$\not=\ZPP$.
\end{corollary}


\begin{corollary}
    If \MCSP~ is $\ZPP$-hard under polynomial time truth-table
    reductions, then \EXP$\not=\ZPP$.
\end{corollary}





\section{Conclusions}
In this paper, we develop an algebraic  method to magnify  the
hardness of sparse sets in nondeterministic classes via a randomized
streaming model. It has a flexible trade off between the sparseness
and time complexity. This shows connection to the major problems to
prove $\NEXP\not=\BPP$. We also prove that if $\MCSP$~is \ZPP-hard,
then $\EXP\not=\ZPP$.

{\bf Acknowledgements:} This research was supported in part by National Science
    Foundation Early Career Award 0845376, and Bensten Fellowship of the
    University of Texas Rio Grande Valley.
    Part of this research was conducted while the author was visiting the School of Computer  Science and Technology of Hengyang Normal University in the summer of 2019 and was supported by National Natural Science Foundation of China 61772179.







\end{document}

\end{document}

%% file: mac90.tex
\newcounter{savenumi}

\newtheorem{theoremfoo}{Theorem}
\newenvironment{theorem}{\pagebreak[1]\begin{theoremfoo}}{\end{theoremfoo}}

\newtheorem{propositionfoo}[theoremfoo]{Proposition}
\newenvironment{proposition}{\pagebreak[1]\begin{propositionfoo}}{\end{propositionfoo}}
\newtheorem{lemmafoo}[theoremfoo]{Lemma}
\newenvironment{lemma}{\pagebreak[1]\begin{lemmafoo}}{\end{lemmafoo}}
\newtheorem{conjecturefoo}[theoremfoo]{Conjecture}

\newtheorem{corollaryfoo}[theoremfoo]{Corollary}
\newenvironment{corollary}{\pagebreak[1]\begin{corollaryfoo}}{\end{corollaryfoo}}
\newtheorem{exercisefoo}{Exercise}

\newtheorem{openfoo}[theoremfoo]{Question}

\newtheorem{nttn}[theoremfoo]{Notation}

\newtheorem{dfntn}[theoremfoo]{Definition}
\newenvironment{definition}{\pagebreak[1]\begin{dfntn}\rm}{\end{dfntn}}

\newenvironment{proof}
    {\pagebreak[1]{\narrower\noindent {\bf Proof:\quad\nopagebreak}}}{\QED}

\newcommand{\dash}{{\rm\mbox{-}}}

\newcommand{\ceiling}[1]{\left\lceil#1\right\rceil}

\newcommand{\DTIME}{{\rm DTIME}}
\newcommand{\PH}{{\rm PH}}

\newcommand{\NTIME}{{\rm NTIME}}

\newcommand{\poly}{{\rm poly}}
\newcommand{\NP}{{\rm NP}}

\renewcommand{\P}{{\rm P}}      
\newcommand{\coNP}{{\co\NP}}
\def\nre.{$n$\/-r.e.}
\newcommand{\scrod}{\quad\nopagebreak}

\newcommand{\co}{{\rm co}\dash}

\newcommand{\EXP}{{\rm EXP}}
\newcommand{\NEXP}{{\rm NEXP}}

\newtheorem{factfoo}[theoremfoo]{Fact}

\newcommand{\BPP}{{\rm BPP}}

\newtheorem{propertyfoo}[theoremfoo]{Property}

\makeatletter

\def\@makechapterhead#1{ \vspace*{50pt} { \parindent 0pt \raggedright 
 \ifnum \c@secnumdepth >\m@ne \huge\bf \@chapapp{} \thechapter. \par 
 \vskip 20pt \fi \Huge \bf #1\par 
 \nobreak \vskip 40pt } }

\def\@sect#1#2#3#4#5#6[#7]#8{\ifnum #2>\c@secnumdepth
     \def\@svsec{}\else 
     \refstepcounter{#1}\edef\@svsec{\csname the#1\endcsname.\hskip 1em }\fi
     \@tempskipa #5\relax
      \ifdim \@tempskipa>\z@ 
        \begingroup #6\relax
          \@hangfrom{\hskip #3\relax\@svsec}{\interlinepenalty \@M #8\par}
        \endgroup
       \csname #1mark\endcsname{#7}\addcontentsline
         {toc}{#1}{\ifnum #2>\c@secnumdepth \else
                      \protect\numberline{\csname the#1\endcsname}\fi
                    #7}\else
        \def\@svsechd{#6\hskip #3\@svsec #8\csname #1mark\endcsname
                      {#7}\addcontentsline
                           {toc}{#1}{\ifnum #2>\c@secnumdepth \else
                             \protect\numberline{\csname the#1\endcsname}\fi
                       #7}}\fi
     \@xsect{#5}}

\def\@begintheorem#1#2{\it \trivlist \item[\hskip \labelsep{\bf #1\ #2.}]}

\def\@opargbegintheorem#1#2#3{\it \trivlist
      \item[\hskip \labelsep{\bf #1\ #2\ (#3).}]}

\makeatother

\newcommand{\ppoly}{\P/\poly}

\newif\ifshortconferences
\shortconferencesfalse
\newif\ifmediumconferences
\mediumconferencesfalse

\def\ending#1{{\count1=#1\relax
\count2=\count1
\divide\count2 by 100
\multiply\count2 by 100
\advance\count1 by -\count2
\ifnum\count1=11
th%
\else \ifnum\count1=12
th%
\else \ifnum\count1=13
th%
\else 
\count2=\count1
\divide\count1 by 10
\multiply\count1 by 10
\advance\count2 by -\count1
\ifnum\count2=1
st%
\else \ifnum\count2=2
nd%
\else \ifnum\count2=3
rd%
\else th%
\fi\fi\fi\fi\fi\fi
}}

\def\Proceedingsofthe{\ifshortconferences Proc.\else\ifmediumconferences Proc.\else Proceedings of the\fi\fi}

\newcounter{confnum}

\def\conf#1#2{%
\setcounter{confnum}{#2}%
\addtocounter{confnum}{-\csname #1zero\endcsname}%
\ifnum\value{confnum}=1%
\expandafter\ifx\csname #1One\endcsname\relax%
\Proceedingsofthe\ \arabic{confnum}\ending{\value{confnum}}\ \csname #1name\endcsname%
\else \csname #1One\endcsname\fi%
\else%
\Proceedingsofthe\
\arabic{confnum}\ending{\value{confnum}}\ \csname #1name\endcsname\fi}

\def\qsym{\vrule width0.7ex height0.9em depth0ex}

\newif\ifqed\qedtrue

\def\noqed{\global\qedfalse}

\def\qed{\ifqed{\penalty1000\unskip\nobreak\hfil\penalty50
\hskip2em\hbox{}\nobreak\hfil\qsym
\parfillskip=0pt \finalhyphendemerits=0\par\medskip}\fi\global\qedtrue}

\makeatletter
\def\eqnqed{\noqed
	\def\@tempa{equation}
	\ifx\@tempa\@currenvir\def\@eqnnum{\qsym}%
	\addtocounter{equation}{-1}\else%
    \def\@@eqncr{\let\@tempa\relax
    \ifcase\@eqcnt \def\@tempa{& & &}\or \def\@tempa{& &}%
      \else \def\@tempa{&}\fi
     \@tempa {\def\@eqnnum{{\qsym}}\@eqnnum}
     \global\@eqnswtrue\global\@eqcnt\z@\cr}\fi}

\def\eqnlabel#1#2{\if@filesw {\let\thepage\relax%
   \def\protect{\noexpand\noexpand\noexpand}%
   \edef\@tempa{\write\@auxout{\string
      \newlabel{#2}{{{#1}}{\thepage}}}}%
   \expandafter}\@tempa%
   \if@nobreak \ifvmode\nobreak\fi\fi\fi%
	\def\@tempa{equation}
	\ifx\@tempa\@currenvir\def\theequation{{#1}}%
	\addtocounter{equation}{-1}\else%
    \def\@@eqncr{\let\@tempa\relax
    \ifcase\@eqcnt \def\@tempa{& & &}\or \def\@tempa{& &}%
      \else \def\@tempa{&}\fi
     \@tempa {\def\@eqnnum{{#1}}\@eqnnum}
     \global\@eqnswtrue\global\@eqcnt\z@\cr}\fi}

\makeatother

\def\QED{\qed}

\makeatother

